\documentclass[11pt]{article}

\usepackage{chao}
\usepackage{graphicx}
\usepackage{caption}
\usepackage{subcaption}
\newcommand{\sumdeg}{\text{\sf sum-deg}}

\newcommand{\cD}{\mathcal{D}}

\newcommand{\mypara}[1]{\medskip \noindent {\bf #1}}

\begin{document}

\title{Computing minimum cuts in hypergraphs\thanks{Department of 
Computer Science, University of Illinois, Urbana, IL 61801. 
\texttt{\{chekuri,chaoxu3\}@illinois.edu}. Work on this paper supported in part by NSF 
grant CCF-1526799.}}
\author{Chandra Chekuri \and Chao Xu}
\date{}

\maketitle

\begin{abstract}
  We study algorithmic and structural aspects of connectivity in
  hypergraphs. Given a hypergraph $H=(V,E)$ with $n = |V|$, $m = |E|$
  and $p = \sum_{e \in E} |e|$ the fastest known algorithm to compute a
  global minimum cut in $H$ runs in $O(np)$ time for the uncapacitated
  case, and in $O(np + n^2 \log n)$ time for the capacitated case. We show
  the following new results.
  \begin{itemize}
  \item Given an uncapacitated hypergraph $H$ and an integer $k$ we
    describe an algorithm that runs in $O(p)$ time to find a subhypergraph
    $H'$ with sum of degrees $O(kn)$ that preserves all
    edge-connectivities up to $k$ (a $k$-sparsifier). This generalizes
    the corresponding result of Nagamochi and Ibaraki from graphs to
    hypergraphs. Using this sparsification we obtain an $O(p +
    \lambda n^2)$ time algorithm for computing a global minimum cut of
    $H$ where $\lambda$ is the minimum cut value.
  \item We generalize Matula's argument for graphs to hypergraphs and
    obtain a $(2+\e)$-approximation to the global minimum cut in a
    capacitated hypergraph in $O(\frac{1}{\e} (p \log n + n \log^2 n))$
    time, and in in $O(p/\e)$ time for uncapacitated hypergraphs.
  \item We show that a hypercactus representation of {\em all} the
    global minimum cuts of a capacitated hypergraph can be computed in
    $O(np + n^2 \log n)$ time and $O(p)$ space.
  \end{itemize}
  Our results build upon properties of vertex orderings that were
  inspired by the maximum adjacency ordering for graphs due to
  Nagamochi and Ibaraki.  Unlike graphs we observe that there are
  several different orderings for hypergraphs which yield different
  insights.
\end{abstract}


\section{Introduction}
\label{sec:intro}
We consider algorithmic and structural aspects of connectivity in
hypergraphs. A hypergraph $H=(V,E)$ consists of a finite vertex set
$V$ and a set of hyperedges $E$ where each edge $e$ is a subset of
vertices. Undirected loopless graphs are a special case of hypergraphs
where all edges are sets of size two. For the most part we use $n$ to
denote the number of vertices $|V|$, $m$ to denote the number of edges
$|E|$, and $p$ to denote $\sum_{e \in E} |e|$. Note that $p = \sum_{v
  \in V} \mbox{deg}(v)$ where $\mbox{deg}(v)$ is the degree of $v$
(the number of hyperedges that contain $v$).  We observe that $p$ is
the natural representation size of a connected hypergraph, and $p$ is the number of
edges in the standard representation of $H$ as a bipartite graph $G_H
= (V \cup E, F)$ with $F = \{ (v,e) | v \in V, e \in E, v \in e\}$.  A
number of results on hypergraphs assume that the maximum edge size,
often called the {\em rank}, is a fixed constant $r$. In this paper
our focus is on general hypergraphs without assumptions on $r$.

Hypergraphs arise in a number of settings in both theory and practice.
Some of the most basic algorithmic questions regarding hypergraphs
have to do with connectivity and cuts.  Given a hypergraph $H=(V,E)$,
let $\delta_H(S)$ denote the set of edges that intersect both $S$ and
$V \setminus S$.  It is well-known that the set function
$|\delta_H(S)|$ defines a symmetric submodular set function over the
ground set $V$.  The connectivity (or the global minimum cut value) of
a hypergraph $H$, denoted by $\lambda(H)$, is defined as
$\min_{\emptyset \subsetneq S \subsetneq V} |\delta_H(S)|$;
equivalently, it is the minimum number of edges that need to be
removed such that $H$ is disconnected.  For distinct nodes $s,t \in V$
we denote by $\lambda_H(s,t)$ (or some times by $\lambda(s,t;H)$) the
edge-connectivity between $s$ and $t$ in $H$ which is defined as
$\min_{S \subset V, |S \cap \{s,t\}| = 1} |\delta_H(S)|$.  These
definitions readily generalize to capacitated hypergraphs where each
edge $e \in E$ has a non-negative capacity $c(e)$ associated with it.
In this paper we study algorithmic and structural questions that arise
in computing $\lambda(H)$. In the sequel we use the term mincut to
refer to the global mincut.

Algorithms for mincuts and $s$-$t$ mincuts in graphs have been
extensively studied. Traditional algorithms for mincut were based on
computing a sequence of $(n-1)$ $s$-$t$ mincuts; $s$-$t$ mincuts are
most efficiently computed via network flow although one can also
compute them via submodular function minimization. The first algorithm
for finding a mincut in an {\em undirected} graph that avoided the use
of flows was due to Nagamochi and Ibaraki \cite{Nagamochi1992}.  They
devised a surprising and influential algorithm based on
maximum-adjacency orderings (MA-ordering) which is an ordering of the
vertices based on a simple greedy rule. An MA-ordering can be computed
in $O(m)$ time for uncapacitated graphs and in $O(m + n\log n)$ time
for capacitated graphs. It has the following interesting property: if
$s$ and $t$ are the last two vertices in the ordering then $\{t\}$ is
an $s$-$t$ mincut. This yields a simple $O(mn + n^2 \log n)$ time
algorithm \cite{StoerW} for computing a mincut in a capacitated graph
and is currently the asymptotically fastest {\em deterministic}
algorithm. MA-orderings have other important structural properties
which lead to several algorithmic and structural results --- many of
these are outlined in \cite{NagamochiI-book}.  Karger devised another
highly influential technique based on random
contractions~\cite{Karger1995} which led to a randomized $O(n^2 \log^3
n)$-time Monte Carlo algorithm for computing a mincut in capacitated
graphs \cite{KargerS96}. Subsequently, using sampling techniques for
cut-sparsification and tree packings, Karger devised a randomized $O(m
\log^3 n)$ time Monte Carlo algorithm \cite{Karger00}. More recently
Kawarabayashi and Thorup \cite{KawarabayashiT15} devised a
deterministic $O(m \log^{12} n)$ time algorithm for {\em simple}
uncapacitated graphs.

What about hypergraphs? A simple and well-known reduction shows that
$\lambda_H(s,t)$ can be computed via $s$-$t$ network flow in the
vertex capacitated bipartite graph $G_H$ associated with $H$. Thus,
using $(n-1)$ network flows one can compute $\lambda(H)$. However,
Queyranne~\cite{Queyranne} showed that the Nagamochi-Ibaraki ordering
approach generalizes to find the mincut of an arbitrary symmetric
submodular function\footnote{For a submodular function $f:2^V
  \rightarrow \mathbb{R}$ the mincut is defined naturally as
  $\min_{\emptyset \subsetneq S \subsetneq V} f(S)$.}. A
specialization of the approach of Queyranne with a careful
implementation leads to a deterministic $O(np + n^2 \log n)$-time
algorithm for capacitated hypergraphs and an $O(np)$-time algorithm
for uncapacitated hypergraphs.  Two other algorithms achieving the
same run-time were obtained by Klimmek and Wagner \cite{Klimmek1996}
and Mak and Wong \cite{MakW00}. Both these algorithms are based on the
Nagamochi and Ibaraki ordering approach. Surprisingly, the orderings
used by these three algorithms can be different for the same
hypergraph even though they are identical for graphs\footnote{This
  observation does not appear to have been explicitly noted in the
  literature.}! We will later show that we can exploit their different
properties in our algorithms.

Apart from the above mentioned results, very little else is known in
the algorithms literature on mincut and related problems in
hypergraphs despite several applications, connections, and theoretical
interest. Recent work has addressed streaming and sketching algorithms
when the rank is small \cite{KoganK15,GuhaMT15}. Our initial
motivation to address these algorithmic questions came from the study
of algorithms for element-connectivity and related problems which are
closely related to hypergraphs --- we refer the reader to the recent
survey \cite{Chekuri-ec-survey15}. In this paper the two main
questions we address are the following.

\begin{itemize}
\item Are there faster (approximation) algorithms for mincut computation
in hypergraphs? 
\item How many distinct mincuts can there be? Can a compact representation
  called the hypercactus that is known to exist \cite{Cheng,FleinerJ} 
  be computed fast?  For graphs it is well-known that there are at most
  $n \choose 2$ mincuts and that there exists a compact $O(n)$-sized data
  structure called the {\em cactus} to represent all of
  them. 
\end{itemize}

\subsection{Overview of Results} 

In this paper we address the preceding questions and obtain several 
new results that we outline below.

\mypara{Sparsification and fast algorithm for small mincuts:} A
$k$-sparsifier of a graph $G=(V,E)$ is a sparse subgraph $G'=(V,E')$
of $G$ that preserves all local connectivities in $G$ up to $k$; that
is $\lambda_{G'}(s,t) \geq \min\{k, \lambda_G(s,t)\}$ for all $s,t \in
V$. Nagamochi and Ibaraki \cite{Nagamochi1992} showed, via
MA-ordering, that a $k$-sparsifier with $O(kn)$ edges can be found in
linear time. In the hypergraph setting, a $k$-sparsifier is a
subhypergraph preserving local connectivity up to $k$. A
$k$-sparsifier with $O(kn)$ edges exists by greedy spanning hypergraph
packing \cite{GuhaMT15}. However, the sum of degrees in the sparsifier
can be $O(kn^2)$. Indeed, any $k$-sparsifier through edge deletion
alone cannot avoid the $n^2$ factor. We consider a more general
operation where we allow {\em trimming} of hyperedges; that is, a
vertex $v \in e$ can be removed from $e$ without $e$ itself being
deleted. Trimming has been used for various connectivity results on
hypergraphs. For example, in studying $k$-partition-connected
hypergraphs, or in extending Edmonds' arborescence packing theorem to
directed hypergraphs \cite{FrankKK2003} (see \cite[Section 7.4.1,
Section 9.4]{Frank2011} for a exposition of the results using the
trimming terminology).

We show that for any hypergraph $H$ on $n$ nodes there is a
$k$-sparsifier $H'$ that preserves all the local connectivities up to
$k$ such that the size of $H'$ in terms of the sum of degrees is
$O(kn)$. In fact the sparsifier has the stronger property that all
cuts are preserved up to $k$: formally, for any $A \subset V$,
$|\delta_{H'}(A)| \ge \min\{k, |\delta_H(A)|\}$.  Moreover such a
sparsifier can be constructed in $O(p)$ time.  This leads to an $O(p +
\lambda n^2)$ time for computing the mincut in an uncapacitated
hypergraph, substantially improving the $O(np)$ time when $\lambda$ is
small and $p$ is large compared to $n$.  Sparsification is of
independent interest and can be used to speed up algorithms for
other cut problems.

\mypara{$(2+\e)$ approximation for global mincut:} Matula
\cite{Matula1993}, building on the properties of MA-ordering, showed
that a $(2+\e)$ approximation for the global mincut of an
uncapacitated graph can be computed in deterministic $O(m/\e)$ time.
The algorithm generalizes to capacitated graphs and runs in
$O(\frac{1}{\e}(m \log n +n\log^2 n))$ time (as mentioned by Karger
\cite{Karger1995}). Although the approximation is less interesting in
light of the randomized $\tilde{O}(m)$ algorithm of Karger
\cite{Karger00}, it is a useful building block that allows one to
deterministically estimate the value of a mincut. For hypergraphs
there was no such approximation known. In fact, the survey
\cite{Chekuri-ec-survey15} observed that a near-linear time randomized
$O(\log n)$-approximation follows from tree packing results, and
raised the question of whether Matula's algorithm can be generalized
to hypergraphs\footnote{We
  use near-linear-time to refer to algorithms that run in time $O(p
  \log^c n)$ for some fixed constant $c$.}.

In this paper we answer the question in the affirmative and obtain
a $(2+\e)$-approximation algorithm for the mincut of a capacitated hypergraph
that runs in near-linear time --- more formally in 
$O(\frac{1}{\e} (p \log n + n \log^2 n))$ time. For a uncapacitated 
hypergraph, the algorithm runs in $O(p/\e)$ time.

\mypara{All mincuts and hypercactus:} Our most technical
contribution is for the problem of finding all the mincuts
in a hypergraph. For any capacitated graph $G$ on $n$ vertices, it is
well-known, originally from the work of Dinitz, Karzanov and Lomonosov
\cite{DinitzKL76}, that there is a compact $O(n)$ sized data
structure, namely a cactus graph, that represents all the mincuts of
$G$.  A cactus is a connected graph in which each edge is in at most
one cycle (can be interpreted as a tree of cycles).  As a byproduct
one also obtains the fact that there are at most $n \choose 2$
distinct mincuts in a graph; Karger's contraction algorithm gives a
very different proof. After a sequence of improvements, there exist
deterministic algorithms to compute a cactus representation of the
mincuts of a graph in $O(mn+n^2\log n)$ time \cite{Nagamochi2003} or
in $O(nm\log (n^2/m))$-time \cite{Fleischer,Gabow2016}.  For
uncapacitated graphs, there is an $O(m+\lambda^2n \log(n/\lambda))$-time
algorithm \cite{Gabow2016}. There is also a
Monte Carlo algorithm that runs in $\tilde{O}(m)$ time \cite{KargerP}
building on the randomized near-linear time algorithm of
Karger~\cite{Karger00}. In effect, the time to compute the cactus
representation is the same as the time to compute the global mincut!
We note, however, that all the algorithms are fairly complicated, in
particular the deterministic algorithms.

The situation for hypergraphs is not as straight forward. First, how
many distinct mincuts can there be? Consider the example of a
hypergraph $H=(V,E)$ with $n$ nodes and a single hyperedge containing
all the nodes.  Then it is clear that every $S \subset V$ with $1 \le
|S| < |V|$ defines a mincut and hence there are exponentially
many. However, all of them correspond to the same edge-set. A natural
question that arises is whether the number of distinct mincuts in
terms of the edge-sets is small.  Indeed, one can show that it is at
most $n \choose 2$. However, this fact does not seem to have been
explicitly mentioned in the literature although it was known to some
experts. It is a relatively simple consequence of fundamental
decomposition theorems of Cunningham and Edmonds \cite{CunninghamE},
Fujishige \cite{Fujishige}, and Cunningham \cite{Cunningham} on
submodular functions from the early 1980s. Cheng, building on
Cunningham's work \cite{Cunningham}, explicitly showed that the
mincuts of a hypergraph admit a compact hypercactus structure.  Later
Fleiner and Jordan \cite{FleinerJ} showed that such a structure exists
for any symmetric submodular function defined over crossing families.
However, these papers were not concerned with algorithmic
considerations.

In this paper we show that the hypercactus representation of the
mincuts of a hypergraph, a compact $O(n)$ sized data structure, can be
computed in $O(np + n^2 \log n)$ time and $O(p)$ space. This matches
the time to compute a single mincut. The known algorithms for cactus
construction on graphs are quite involved and directly construct the
cactus.  We take a different approach.  We use the structural theory
developed in \cite{Cunningham,Cheng} to build the canonical
decomposition of a hypergraph which then allows one to build a
hypercactus easily. The algorithmic step needed for constructing the
canonical decomposition is conceptually simpler and relies on an
efficient algorithm for finding a non-trivial mincut (one in which
both sides have at least two nodes) in a hypergraph $H$ if there is
one. Our main technical contribution is to show that there is an
algorithm for finding a slight weakening of this problem that runs in
$O(p + n \log n)$ time. Interestingly, we establish this via the
ordering from the paper of \cite{MakW00}.  Our algorithm yields a
conceptually simple algorithm for graphs as well and serves to
highlight the power of the decomposition theory for graphs and submodular
functions~\cite{CunninghamE,Fujishige,Cunningham}. 

\subsection{Other Related Work} In a recent work Kogan and
Krauthgamer~\cite{KoganK15} examined the properties of random
contraction algorithm of Karger when applied to hypergraphs. They
showed that if the rank of the hypergraph is $r$ then the number of
$\alpha$-mincuts for $\alpha \ge 1$ is at most $O(2^{\alpha r}
n^{2\alpha})$ which is a substantial improvement over a naive analysis
that would give a bound of $O(n^{r \alpha})$. The exponential
dependence on $r$ is necessary. They also showed cut-sparsification
results ala Benczur and Karger's result for graphs
\cite{BenczurK15}. In particular, given a $n$-vertex capacitated hypergraph
$H=(V,E)$ of rank $r$ they show that there is a
capacitated hypergraph $H'=(V,E')$ with $O(\frac{n}{\e^2}(r + \log
n))$ edges such that every cut capacity in $H$ is preserved to within
a $(1\pm \e)$ factor in $H'$. Aissi et al.\ \cite{AissiMMQ15}
considered parametric mincuts in graphs and hypergraphs of fixed rank
and obtained polynomial bounds on the number of distinct mincuts.

Hypergraph cuts have also been studied in the context of $k$-way cuts.
Here the goal is to partition the vertex set $V$ into $k$ non-empty
sets so as to minimize the number of hyperedges crossing the
partition. For $k\le 3$ a polynomial time algorithm is known
\cite{Xiao10} while the complexity is unknown for fixed $k \ge 4$. The
problem is NP-Complete when $k$ is part of the input even for graphs
\cite{GoldschmidtH94}. Fukunaga~\cite{Fukunaga13} obtained a
polynomial-time algorithm for $k$-way cut when $k$ and the rank $r$
are fixed; this generalizes the result the polynomial-time algorithm
for graphs \cite{GoldschmidtH94,Thorup08}.  Karger's contraction
algorithm also yields a randomized algorithm when $k$ and the rank $r$
are fixed.  When $k$ is part of the input, $k$-way cut in graphs
admits a $2(1-1/k)$-approximation \cite{SaranV95}. This immediately
yields a $2r(1-1/k)$-approximation for hypergraphs. If $r$ is not
fixed and $k$ is part of the input, it was recently shown
\cite{ChekuriL15} that the approximability of the $k$-way cut problem
is related to that of the $k$-densest subgraph problem.

Hypergraph cuts arise in several other contexts with terminals such as
the $s$-$t$ cut problem or its generalizations such as multi-terminal
cut problem and the multicut problem. In some of these problems one
can reduce the hypergraph cut problem to a node-capacitated undirected
graph cut problem and vice-versa.

\section{Preliminaries}
\label{sec:prelim}
A hypergraph $H=(V,E)$ is capacitated if there is a non-negative edge
capacity function $c:E\to \R_+$ associated with it. If all capacities
are $1$ we call the hypergraph uncapacitated; we allow multiple copies
of an edge in the uncapacitated case.  A cut $(S,V-S)$ is a
bipartition of the vertices, where $S$ and $V-S$ are both
non-empty. We will abuse the notation and call a set $S$ a cut to mean
the cut $(S,V-S)$.  Two sets $A$ and $B$ \emph{cross} if $A\cap B$,
$A\setminus B$ and $B\setminus A$ are all non-empty. For $S \subseteq
V$, $\delta_H(S)$ is defined to be the set of all edges in $E(H)$ that
have an endpoint in both $S$ and $V-S$.  We will drop $H$ from the
notation if the hypergraph is clear from the context. The capacity of
the cut $S$, denoted by $c(S)$, is defined to be $\sum_{ e\in
  \delta(S)} c(e)$. $\lambda(H)$ is the connectivity of $H$ and is
defined as $\min_{\emptyset \subsetneq S \subsetneq V} c(S)$.  A cut $S$
is a mincut of $H$ if $c(S) = \lambda(H)$.  A set of edges $F$ is an
edge-cut-set if $F=\delta(S)$ for some cut $S$. A set of edges is a
min edge-cut-set if $F=\delta(S)$, where $S$ is a mincut. For distinct
vertices $s,t \in V(H)$, an $s$-$t$ cut is a cut $S$ such that $|S
\cap \{s,t\}| = 1$.  $\lambda(s,t;H)$ is the value of the $s$-$t$
mincut(minimum $s$-$t$ cut) in $H$. $\sumdeg(H)$ is defined as
$\sum_{e\in E} |e|$ in the uncapacitated case and as
$\sum_{e\in E} |e| c(e)$ in the capacitated case.

For pairwise disjoint vertex subsets $A_1,\ldots,A_k$,
$E(A_1,\ldots,A_k;H) = \{e \mid e \cap A_i \neq \emptyset, 1 \le i \le
k\}$ is the set of edges that have an end point in each of the sets
$A_1,\ldots,A_k$.  $d(A_1,\ldots,A_k;H) = \sum_{e\in
  E(A_1,\ldots,A_k;H)} c(e)$ denotes the total capacity of the edges
in $E(A_1,\ldots,A_k;H)$. A related quantity for two disjoint sets $A$
and $B$ is $d'(A,B;H) = \sum_{e\in E(A,B;H), e\subset A\cup B} c(e)$
where only edges completely contained in $A \cup B$ are considered.
As before, if $H$ is clear from the context we drop it from the notation.

Removing a vertex $v \in e$ from $e$ is called {\em trimming} $e$~\cite{Frank2011}. A hypergraph $H'=(V',E')$ is a
\emph{subhypergraph} of $H=(V,E)$ if $V'\subset V$ and there is a
injection $\phi: E\to E'$ where $\phi(e) \subseteq e$. Thus a
subhypergraph of $H$ is obtained be deleting vertices and edges and
trimming edges\footnote{The standard definition of
    subhypergraph does not allow trimming, but our definition is
    natural for sparsification.}.

For simplicity, given hypergraph $H=(V,E)$, we use $n$ as the number
of vertices, $m$ as the number of edges, and $p = \sum_{e\in E} |e|$
as the sum of degrees. 

\mypara{Equivalent digraph:} $s$-$t$ mincut in a hypergraph $H$ can be
computed via an $s$-$t$ maximum flow in an associated capacitated digraph
(see \cite{Lawler73}) $\vec{H}=(\vec{V},\vec{E})$ that we call the
equivalent digraph. $\vec{H}=(\vec{V},\vec{E})$ is defined as follows:

\begin{enumerate}
\item $\vec{V}=V\cup E^+\cup E^-$, where $E^+=\{e^+|e\in E\}$ and $E^-=\{e^-|e\in E\}$.
\item If $v\in e$ for $v\in V$ and $e\in E$ then $(v,e^-)$ and $(e^+,v)$ are in $\vec{E}$ with infinite capacity.
\item For each $e\in E$, $(e^-,e^+)\in \vec{E}$ has capacity equal to $c(e)$. 
\end{enumerate}

For any pair $s,t \in V(H)$, there is bijection between the finite capacity
$s$-$t$ cuts in $\vec{H}$ and $s$-$t$ cuts in $H$. We omit further details of
this simple fact.

\mypara{Cactus and hypercactus:} A \emph{cactus} is a graph in which
every edge is in at most one cycle. A \emph{hypercactus} is a
hypergraph obtained by a sequence of \emph{hyperedge insertions}
starting from a cactus. A hyperedge insertion is defined as follows. A
vertex $v$ in a hypergraph with degree at least $3$ and only incident
to edges of rank $2$, say $vv_1,\ldots,vv_k$ is called a $v$-star. A
hyperedge insertion replaces a $v$-star by deleting $v$, adding new
vertices $x_1,x_2,\ldots, x_k$, adding new edges
$\set{x_1,v_1},\set{x_2,v_2},\ldots,\set{x_k,v_k}$ and a new hyperedge
$\set{x_1,x_2,\ldots,x_k}$. See \autoref{fig:cactus} for examples.

\begin{figure}
\begin{center}
\begin{subfigure}[b]{0.4\textwidth}
        \includegraphics[width=\textwidth]{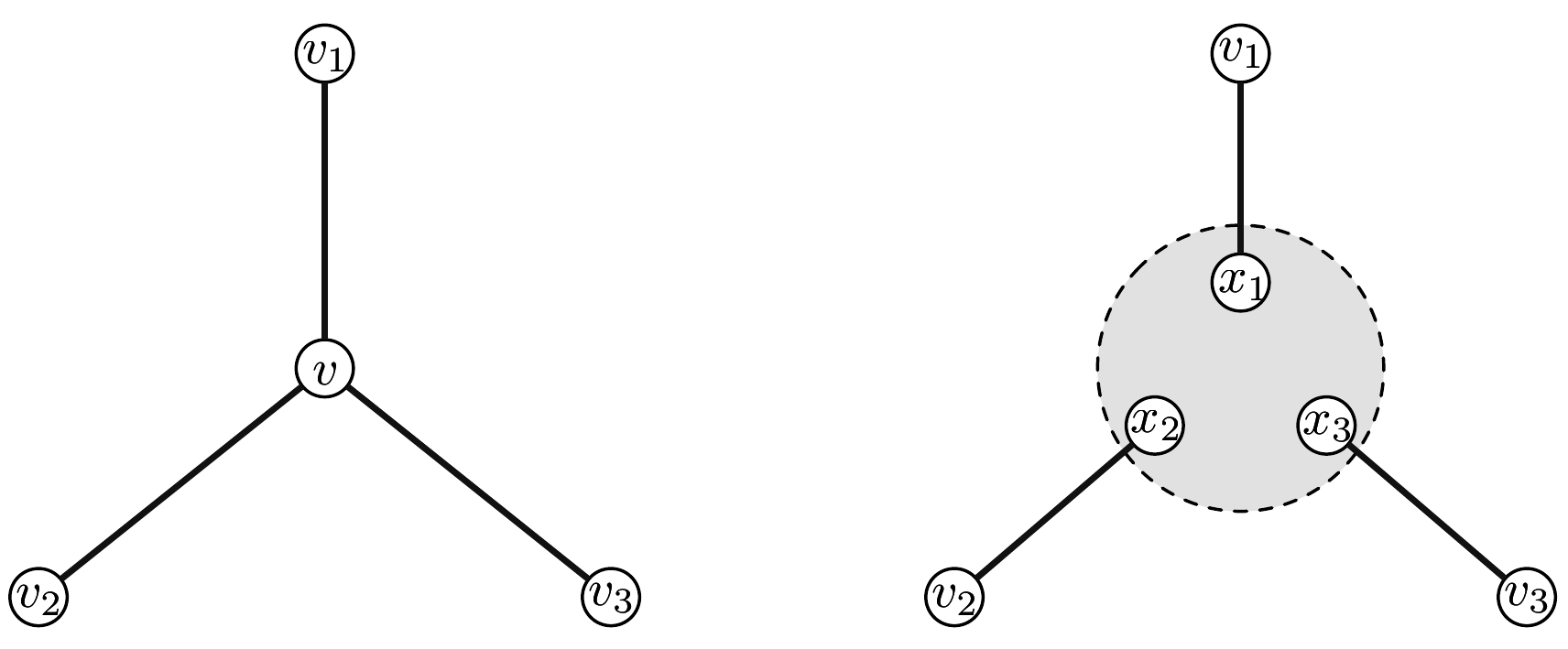}
        \caption{Example of hyperedge insertion operation on vertex $v$}
        \label{fig:hyperedgeinsert}
    \end{subfigure}
    ~ 
    \begin{subfigure}[b]{0.25\textwidth}
        \includegraphics[width=\textwidth]{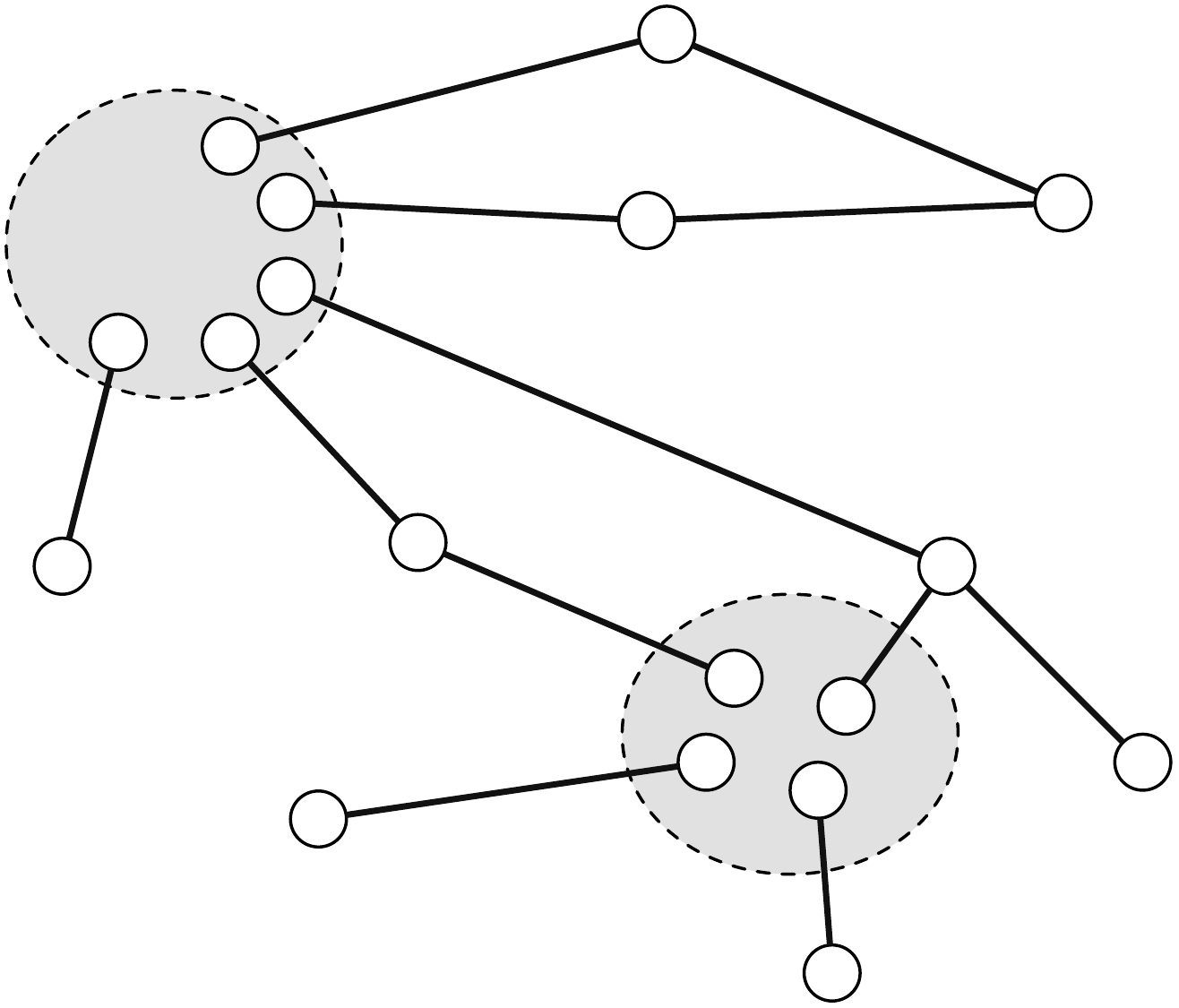}
        \caption{A hypercactus}
        \label{fig:hypercactus}
    \end{subfigure}
\end{center}

\caption{
    Examples of hypercactus. Grey regions are hyperedges. 
}
\label{fig:cactus}
\end{figure}

\subsection{Vertex orderings for  hypergraphs}
We work with several vertex orderings defined for hypergraphs.
Given a hypergraph $H=(V,E)$ and an ordering of the vertices
$v_1,\ldots,v_n$, several other orderings and quantities are induced
by such an ordering. The \emph{head} of an edge $h(e)$, defined as
$v_{\min \set{j|v_j \in e}}$, is the first vertex of $e$ in the
ordering. An ordering of the edges $e_1,\ldots,e_m$ is called \emph{head
  ordering}, if $\min \set{j |v_j\in e_i\} \leq \min \{j |v_j\in
  e_{i+1}}$. An edge $e$ is a \emph{backward edge} of $v$ if $v\in e$ and
$h(e) \neq v$. The head of a backward edge incident to $v$ comes
before $v$ in the vertex order. $V_i=\{v_1,\ldots,v_i\}$ are the first
$i$ vertices in the ordering.

A pair of vertices $(s,t)$ is called a \emph{pendant pair} if
$\set{t}$ is a minimum $s$-$t$ cut. There are three algorithms for
computing a hypergraph mincut following the Nagamochi-Ibaraki approach
of finding a pendant pair, contracting the pair, and recursing. All
three algorithms find a pendant pair by computing an ordering of the
vertices and showing that the last two vertices form a pendant pair.
We describe these orderings below.

\begin{definition}
An ordering of vertices $v_1,\ldots,v_n$ is called 
\begin{enumerate}
\item a \emph{maximum adjacency ordering} or \emph{MA-ordering} if for all $1\leq i < j\leq n$, $d(V_{i-1},v_i) \geq d(V_{i-1},v_j)$.
\item a \emph{tight ordering} if for all $1\leq i < j\leq n$,
$d'(V_{i-1},v_i) \geq d'(V_{i-1},v_j)$.
\item a \emph{Queyranne ordering} if for all $1\leq i < j\leq n$,
$d(V_{i-1},v_i)+d'(V_{i-1},v_i) \geq d(V_{i-1},v_j)+d'(V_{i-1},v_j)$.
\end{enumerate}
\end{definition}

In graphs the three orderings coincide if the starting vertex is the
same and ties are broken in the same way. However, they can be
different in hypergraphs.  As an example, consider a hypergraph with
vertices $a,x,y,z$ and four edges with capacities as follows:
$c(\set{a,x})=4, c(\set{a,y})=3, c(\set{a,x,z})=4$ and
$c(\set{a,y,z})=8$. Capacities can be avoided by creating multiple
copies of an edge.  Consider orderings starting with $a$. It can be
verified that the second vertex has to be $x,y$ and $z$ for tight,
Queyranne, and MA-ordering respectively which shows that they have to
be distinct.

Klimmek and Wagner used the MA-ordering \cite{Klimmek1996}.  Mak and
Wong used the tight ordering\cite{MakW00}. Queyranne defined an
ordering for symmetric submodular functions \cite{Queyranne} which
when specialized to cut functions of hypergraphs is the one we define; we
omit a formal proof of this observation. All three orderings can
be computed in $O(p+n\log n)$ time for capacitated hypergraphs, and in
$O(p)$ time for uncapacitated hypergraphs.  We do not use the
Queyranne ordering in this paper but rely on the other two. We state 
a lemma that summarizes the crucial property that the three orderings
share regarding the mincut between the last two vertices in the ordering.

\begin{lemma}
\label{lem:tightorderlasttwo}
Let $v_1,\ldots,v_n$ be a MA-ordering, a tight ordering or a Queyranne
ordering of a hypergraph, then $\{v_n\}$ is a $v_{n-1}$-$v_n$ mincut
and $\lambda(v_{n-1},v_n) = d'(V_{n-1},v_n) = d(V_{n-1},v_n)$.
\end{lemma}

A stronger property holds for the MA-ordering.
\begin{lemma}
\label{lem:connectivityMAorder}
Let $v_1,\ldots,v_n$ be an MA-ordering of hypergraph $H$, then $\lambda(v_{i-1},v_i)\geq d(V_{i-1},v_i)$ for all $1<i\leq n$. 
\end{lemma}
\begin{proof}
Let $H_i=(V_i,E_i)$, where $E_i = \set{e\cap V_i |e\in E}$. One can check
that $v_1,\ldots,v_i$ is an MA-ordering of $H_i$, and hence,
\[
\lambda(v_{i-1},v_i;H) \geq \lambda(v_{i-1},v_i;H_i) = d(V_{i-1},v_i;H_i) = d(V_{i-1},v_i;H).
\]
\qed\end{proof}

\section{Sparsification and faster mincut algorithm for small $\lambda$}
\label{sec:sparsification}
This section shows that a well-known sparsification result for graphs
can be generalized to hypergraphs. The hypergraphs in this section are
uncapacitated.

Given an uncapacitated hypergraph $H$ and a non-negative integer $k$,
the goal of sparsification is to find a sparse subhypergraph $H'$ of
$H$ such that $\lambda_{H'}(s,t) \geq \min(k,\lambda_{H}(s,t))$ for
all $s,t\in V$. We call such a subhypergraph a $k$-sparsifier.  It is
known that there exists a subhypergraph with $O(kn)$ edges through
edge deletions~\cite{GuhaMT15}. The fact that such a sparsifier exists
is not hard to prove. One can generalize the forest decomposition
technique for graphs \cite{Nagamochi1992} in a straight-forward
way. However, the sum of degrees of the resulting sparsifier could
still be large. Indeed, there might not exist any $k$-sparsifier with
sum of degree $O(kn)$ through edge deletion alone. Consider the
following example. Let $H=(V,E)$ be a hypergraph on $n$ vertices and
$n/2-1$ edges (assume $n$ is even) where $E=
\set{e_1,\ldots,e_{n/2-1}}$ and $e_i = \set{v_i, v_{n/2},\ldots,v_n}$
for $1 \le i<n/2$. Any connected subhypergraph of $H$ has to contain
all the edges and thus, even for $k=1$, the sum of degrees is
$\Omega(n^2)$.

However, if trimming of edges is allowed, we can prove the following
stronger result.
\begin{theorem}
  \label{thm:sparsification}
  Let $H=(V,E)$ be a hypergraph on $n$ vertices and $m$ edges with
  $\sumdeg(H) = p$. There is a data structure that can be
    created in $O(p)$ time, such that for any given non-negative
    integer $k$, it can return a $k$-sparsifier $H'$ of $H$ in
    $O(\sumdeg(H'))$ time with the property that $\sumdeg(H')=O(kn)$.
\end{theorem}

Our proof is an adaptation of that of Frank, Ibaraki and Nagamochi
\cite{Frank1993} for the graph version of the sparsification theorem.

Given a hypergraph $H=(V,E)$ consider an MA-ordering $v_1,\ldots,v_n$
and let $e_1,\ldots,e_m$ be the induced head ordering of the
edges. Let $D_k(v)$ to be the first $k$ backward edges of $v$ in the
head ordering, or all the backward edges of $v$ if there are fewer
than $k$ backward edges. For each vertex $v$ and backward edge $e$ of
$v$, we remove $v$ from $e$ if $e\not\in D_k(v)$. The new hypergraph
from this operation is $H_k$. Formally, given $H$ and $k$,
$H_k=(V,E_k)$ is defined as follows. For an edge $e \in E$ let $e'$
denote the edge $\set{v \mid v\in e\in D_k(v) \text{ or } v=h(e)}$;
$E_k$ is defined to be the edge set $\set{e' \mid e \in E, |e'|\geq
  2}$.  It is easy to see that if $j\leq k$, $H_j$ is a subhypergraph
of $H_k$.

We observe that $\sumdeg(H_k) \le 2kn$. Each vertex $v$ is in at most
$k$ backward edges in $H_k$ for a total contribution of at most $kn$
to the sum of degrees, and the remaining contribution of at most $kn$
comes from head of each edge which can be charged to the backward
edges. 

We sketch a data structure that can be created from
  hypergraph $H$ in $O(p)$ time, such that for all $k$, the data
  structure can retrieve $H_k$ in $O(kn)$ time.  First, we compute the
  MA-ordering, which takes $O(p)$ time. Using the MA-ordering, we
  obtain the induced head ordering of the edges, and the head for each
  edge, again in $O(p)$ time; we omit the simple details for this step. For
  each vertex $v$, we can sort all the backedges of $v$ use the head
  ordering in $O(p)$ time as follows: we maintain a queue $Q_v$ for
  each vertex $v$, and inspect the edges one by one in the head
  ordering. When $e$ is inspected, we push $e$ into queue $Q_v$ if $v\in
  e$ and $v$ is not the head of $e$. This completes the preprocessing
  phase for the data structure. To retrieve $H_k$, we maintain a set
  of queues that eventually represent the set of edges $E_k$. For each
  vertex $v$, find the edges in $D_k(v)$, which is exactly the first
  $k$ edges (or all edges if there are fewer than $k$ edges) in
  $Q_v$. For each edge $e\in D_k(v)$, we push $v$ into a queue $Q_e$ (if
  $Q_e$ was not already created, we first create $Q_e$ and push the head
  vertex of $e$ into $Q_e$). At the end of the process, each queue $Q_e$
  contains the vertices of an edge in $E_k$. The running time is
  $O(\sumdeg(H_k)) = O(kn)$ since we only process $D_k(v)$ for each
  $v$.

It remains to show that $H_k$ is a $k$-sparsifier of $H$. In fact, we
will show $H_k$ preserves more than just connectivity, it also
preserves all cuts up to value $k$.  Namely $|\delta_{H_k}(A)|\geq
\min(k, |\delta_H(A)|)$ for all $A\subset V$.

\begin{lemma}
\label{lem:sparseMAordering}
If $v_1,\ldots,v_n$ is an MA-ordering of $H$, and $H_k$ is a
hypergraph obtained from $H$ via the ordering, then
$v_1,\ldots,v_n$ is an MA-ordering of $H_k$.
\end{lemma}
\begin{proof}
  By construction, $d(V_i,v_j;H_k)\leq \min(k,d(V_i,v_j;H))$ for all
  $i<j$.  Consider the first $\min\{k,d(V_i,v_j;H)\}$ edges incident to
  $v_j$ in the head ordering; $v_j$ is not trimmed from them. Hence
  $d(V_i,v_j;H_k)\geq \min\{k,d(V_i,v_j;H)\}$.

  For all $i\leq j$,
  \[
  d(V_{i-1},v_i;H_k) \geq \min\{ k,d(V_{i-1},v_i;H)\} \geq \min\{ k,d(V_{i-1},v_j;H)\} \geq d(V_{i-1},v_j;H_k).
  \]
This establishes that $v_1,\ldots,v_n$ is an MA-ordering for $H_k$.
\qed\end{proof}

For $X \subseteq V$ we define $\gamma(X) = \{e \mid e \cap X
  \neq \emptyset\}$ to be the set of edges that contain at least one
  vertex from $X$. We need a helper lemma below.
\begin{lemma}
\label{lem:setcancel}
Let $H=(V,E)$ be a hypergraph and $A, B \subseteq V$.
For $u\in B$ and $v\in V$, if $E(u,v)\subset \delta(A)\cap \gamma(B)$, then
\[
|\gamma(B)\cap \delta(A)| \geq |\gamma(B-u)\cap \delta(A)| + |E(u,v)\setminus E(B-u,u,v)|.
\]
\end{lemma}
\begin{proof}
  Consider an edge $e\in E(u,v) \setminus E(B-u,u,v)$. We claim
  that $e\not\in \gamma(B-u)$. Indeed, if $e\in \gamma(B-u)$, then $e$ is
  an edge that intersects $B-u, \set{u\}$ and $\{v}$, but then $e\in
  E(B-u,u,v)$. This shows $E(u,v)\setminus E(B-u,u,v)$ is disjoint
  from $\gamma(B-u)$, and therefore disjoint from 
  $\gamma(B-u)\cap \delta(A)$.

We have (i) $\gamma(B-u) \cap \delta(A)\subset \gamma(B) \cap \delta(A)$ since $\gamma(B-u) \subseteq \gamma(B)$, and (ii) $E(u,v) \setminus E(B-u,u,v) \subset \gamma(B) \cap \delta(A)$ by assumption. Since we have argued that
$\gamma(B-u)$ and $E(u,v) \setminus E(B-u,u,v)$ are disjoint,
we have the desired inequality
\[
|\gamma(B)\cap \delta(A)| \geq |\gamma(B-u)\cap \delta(A)| + |E(u,v)\setminus E(B-u,u,v)|.
\]
\qed\end{proof}

\begin{lemma}
\label{lem:maorderbound}
Let $v_1,\ldots,v_n$ be an MA-ordering for $H=(V,E)$.
Then, for all $i<j$ and $A\subset V$ such that $v_i\in A$ and $v_j\not \in A$. $|\gamma(V_{i-1})\cap \delta(A)|\geq d(V_{i-1},v_j)$. 
\end{lemma}
\begin{proof}
  Proof by induction on $(i,j)$ ordered lexicographically.  For the
  base case consider $i=1$ and $j > 1$. Indeed, in this case both
  sides of the inequality are $0$ and the desired inequality holds
  trivially. Assume lemma is true for all $(i',j')$ where $1 \le i'<j'$,
  such that $j'<j$ or $j'=j$ and $i'<i$. We consider two cases.

\noindent
  {\bf Case 1:} $v_{i-1}\in A$. Then because $v_j\not \in A$,
  $E(v_{i-1},v_j) \subset \delta(A)\cap \gamma(V_{i-1})$. We apply
   \autoref{lem:setcancel} with $B = V_{i-1}$ and $u = v_{i-1}$ and
   $v = v_j$ to obtain:
\begin{align*}
~|\gamma(V_{i-1})\cap \delta(A)| \geq&~ |\gamma(V_{i-2})\cap \delta(A)| + |E(v_{i-1},v_j)\setminus E(V_{i-2},v_{i-1},v_j)|\\
                                \geq&~ d(V_{i-2},v_j)+|E(v_{i-1},v_j)\setminus E(V_{i-2},v_{i-1},v_j)| \quad~ \mbox{(apply induction on $(i-1,j)$, $A$)}\\
                                =&~d(V_{i-1},v_j).
\end{align*}

\noindent {\bf Case 2:} $v_{i-1}\not\in A$. Note that $i \ge 2$.
Consider $A' = V \setminus A$. We have $v_{i-1} \in A'$ and $v_i \not
\in A'$; therefore, $E(v_{i-1},v_i) \subseteq \delta(A') \cap
\gamma(V_{i-1})$. Applying \autoref{lem:setcancel} with $B = V_{i-1}$, $u = v_{i-1}$ and $v = v_i$,
\begin{align*}
|\gamma(V_{i-1})\cap \delta(A')| \geq &~|\gamma(V_{i-2}) \cap \delta(A')| + |E(v_{i-1},v_i) \setminus E(V_{i-2},v_{i-1},v_i)|\\
  \geq &~ d(V_{i-2},v_i) + |E(v_{i-1},v_i) \setminus E(V_{i-2},v_{i-1},v_i)| \quad \mbox{(apply induction on $(i-1,i), A'$)}\\
 = &~ d(V_{i-1},v_i) \\
                               \geq&~d(V_{i-1},v_j) \quad \mbox{(from MA-ordering)}.
\end{align*}
Since $\delta(A') = \delta(V\setminus A) = \delta(A)$, 
$|\gamma(V_{i-1})\cap \delta(A')| = |\gamma(V_{i-1})\cap \delta(A)|$.

This finishes the proof.
\qed\end{proof}

Using the preceding lemma we finish the proof that $H_k$ is a $k$-sparsifier. 

\begin{theorem}
  \label{thm:sparsification-cuts}
  For every $A\subset V$, $|\delta_{H_k}(A)| \geq \min(k,|\delta_{H}(A)|)$.
\end{theorem}
\begin{proof}
  By induction on $k$. The statement is clearly true for
  $k=0$. $|\delta_{H_i}(A)|\leq |\delta_{H_k}(A)|$ for all $i\leq k$,
  because $H_i$ is a subhypergraph of $H_k$.

  Consider any $k>0$. If $|\delta_{H}(A)| = k'<k$, then by induction,
  \[
  k'=|\delta_{H_{k'}}(A)|\leq |\delta_{H_{k}}(A)|\leq |\delta_{H}(A)| = k'.
  \]
  
  Therefore, it suffices to only consider $A$ such that
  $|\delta_{H}(A)| \geq k$. We will derive a contradiction 
  assuming that $|\delta_{H_k}(A)| < k$. Since $|\delta_H(A)| \ge k$, there exists
  an edge $e \in E(H)$ such that $e \in \delta_H(A)$ but was either
  trimmed to $e'\in E(H_k)$ such that $e' \not \in \delta_{H_k}(A)$ or
  $e$ is removed completely because it is trimmed to be a
  singleton $\{h(e)\}$. Let $v_i = h(e)$ and without loss of generality
  we can assume $v_i \in A$ (otherwise we can consider $\bar{A}$).
  Since $e'$ does not cross $A$ in $H_k$, there is a $v_j \in
  e \cap \bar{A}$ with $j > i$ (since $v_i$ is the head of $e$) and
  $v_j$ was trimmed from $e$ during the sparsification.

  $D_k(v_j)$ has exactly $k$ edges because backward edge $e$ of $v_j$ is
  not in $D_k(v_j)$. For each $f\in D_k(v_j)$, the trimmed $f'\in
  E(H_k)$ contains both $h(f)=v_\ell$ and $v_j$; we claim that for
  each such $f$, $\ell \le i$ for otherwise $e$ would be ahead of $f$
  in the head order and $v_j$ would be trimmed from $f$ before it is
  trimmed from $e$. From this we obtain that $E(V_{j-1},v_j;H_k) =
  E(V_i,v_j;H_k)$ and hence $d(V_{j-1},v_j;H_k) = d(V_i,v_j;H_k) = k$.
  
   For the remainder of the the proof, we only work with $H_k$ and
   the quantities $d,\delta,E,\gamma$ are with respect to this hypergraph and not $H$.
   We have 

   \[
   k = d(V_{j-1},v_j)= d(V_i,v_j)= d(V_{i-1},v_j) + d(v_i,v_j) - d(V_{i-1},v_i,v_j).
   \]

Hence $d(V_{i-1},v_j)= k-d(v_i,v_j)+d(V_{i-1},v_i,v_j)$. 

From \autoref{lem:sparseMAordering} $v_1,\ldots,v_n$ is an
MA-ordering of $H_k$ as well. Applying \autoref{lem:maorderbound} to
$H_k$, $v_i$ and $v_j$ and $A$, $|\gamma(V_{i-1})\cap \delta(A)| \geq d(V_{i-1},v_j)$. Combining this inequality with the preceding one,

\begin{equation}
  \label{eq:sparse-helper}
|\gamma(V_{i-1})\cap \delta(A)| \geq d(V_{i-1},v_j) =k-d(v_i,v_j)+d(V_{i-1},v_i,v_j).  
\end{equation}

We also observe that $E(V_{i-1},v_i,v_j)\subset E(v_i,v_j) \subset
\delta(A)$ because $v_i\in A$ and $v_j\not\in A$. We obtain a contradiction
by the following set of inequalities:

\begin{align*}
k           &> |\delta(A)| \quad \mbox{(by assumption)}\\
            &\geq |\gamma(V_i)\cap \delta(A)|\\
            & =~|\gamma(V_{i-1})\cap \delta(A)| + |E(v_i,v_j)\setminus E(V_{i-1},v_i,v_j)|\\
           &  =~|\gamma(V_{i-1})\cap \delta(A)| + d(v_i,v_j) - d(V_{i-1},v_i,v_j) \\
            &\geq \left(k-d(v_i,v_j)+d(V_{i-1},v_i,v_j)\right) + d(v_i,v_j) - d(V_{i-1},v_i,v_j) \quad \mbox{(from (\ref{eq:sparse-helper}))}\\
            &= k
\end{align*}
This finishes the proof.
\qed\end{proof}

There are applications where one want a $k$-sparsifier with
$O(kn)$ edges using only deletion. This can be done easily by first
compute $H_k$, and for each edge in $H_k$, replace it with the 
original edge in $H$. 

One can ask whether tight ordering or Queyranne ordering would also
lead to sparsification. We observe that they do not work if the only
modification is the input ordering, and $H_k$ is constructed the same
way through the head ordering of the edges.

For tight ordering, consider $H=(\set{0,1,2,3},E)$, where $E=
\set{\set{0, 1, 2},\set{0, 2, 3}, \set{1, 2}}$. $0,1,2,3$ is a tight
ordering. $H_2$ is $H$ with edge $\set{1,2}$
removed. $\lambda_{H_2}(1,2)=1<2=\lambda_H(1,2)$.

For Queyranne ordering, consider $H=(\set{0,\ldots,4},E)$, where
\[
E= \set{\set{0,1,2},\set{0,1,2,3},\set{0,1,3,4},\set{1,3,4},\set{2,3}}.
\]
$0,1,2,3,4$ is a Queyranne ordering. $H_3$ is all the edges except the
edge $\set{2,3}$.  We have $\lambda_{H_3}(2,3)=2<3=\lambda_{H}(2,3)$.

There may be other ways to obtain a sparsifier using these
orderings but we have not explored this.

\mypara{Algorithmic applications:} Computing connectivity in uncapacitated
hypergraphs can be sped up by first sparsifying the given hypergraph and then
running a standard algorithm on the sparsifier. This is especially useful
when one is interested in small values of connectivity. For global
mincut we obtain the following theorem.

\begin{theorem}
  The global mincut of a uncapacitated hypergraph $H$ with $n$ vertices
  and $\sumdeg(H) = p$ can be computed in $O(p+\lambda n^2)$ time,
  where $\lambda$ is the value of the mincut of $H$.
\end{theorem}
\begin{proof}
  Using \autoref{thm:sparsification}, we first compute a
  data structure in $O(p)$ time that allows us to retrieve $H_k$ 
  in $O(kn)$ time for any given $k$.
  Suppose we knew a number $k$ such that $k \ge \lambda$ and $k = O(\lambda)$.
  We can compute the $k$-sparsifier $H_k$ in 
  $O(kn)$ time and compute $\lambda(H_k) = \lambda(H)$ in $O(\lambda n^2)$ time
  using one of the known algorithms since $\sumdeg(H_k) = O(\lambda n)$.
  To find $k$ we apply exponential search for the
  smallest $i$ such that $2^i>\lambda$. Each search computes
  hypergraph mincut on $H_{2^i}$, which takes $O(2^i n^2)$ time. For
  any $k> \lambda$, the value of the mincut on the $k$-sparsifier
  equals to $\lambda$. Therefore, the search stops when the value of
  mincut of $H_{2^i}$ is strictly smaller than
  $2^i$. The total time for the computation is
  $O(p+\sum_{i=1}^{1+\ceil{\log \lambda}} 2^i n^2) =  O(p+\lambda n^2)$.
\qed\end{proof}

A similar idea can be used to compute the edge-connectivity in $H$
between some given pair $s,t$. An algorithm for $s$-$t$ connectivity
that runs in time $T(n,m,p)$ on a hypergraph with $n$ nodes, $m$ edges
and sum of degrees $p$ can be sped up to $T(n,m,\lambda(s,t) n)$.
Sparsification can also help in computing $\alpha$-approximate mincuts
for $\alpha > 1$ via the stronger property guaranteed by
\autoref{thm:sparsification-cuts}.

\section{Canonical decomposition and Hypercactus Representation}
\label{sec:decomposition}
In this section, we are interested in finding a canonical
decomposition of a capacitated hypergraph which captures, in a compact
way, information on all the mincuts. Cunningham~\cite{Cunningham}
proved that such decomposition exists for an arbitrary non-negative
submodular function, following previous work by Cunningham and
Edmonds~\cite{CunninghamE} and Fujishige \cite{Fujishige}.  Cheng
\cite{Cheng} showed that the canonical decomposition can be used to
efficiently and relatively easily build a hypercactus representation,
and later Fleiner and Jordan \cite{FleinerJ} showed a similar result
for arbitrary symmetric submodular functions.  In a sense, one can
view the canonical decomposition as a more fundamental object since it
has uniqueness properties while cactus and hypercactus representations
are not necessarily unique.

As noted already by Cunningham, the key tool needed to build a
canonical decomposition is an algorithm to find a non-trivial
mincut. Here we show an efficient algorithm for finding such a mincut
in a hypergraph, and then use it to build a canonical
decomposition. We can then construct a hypercactus from the canonical
decomposition as shown in \cite{Cheng}. We believe that this approach
is easier to understand and conceptually simpler than the existing
deterministic cactus construction algorithms for graphs that build the
cactus directly.

A cut is \emph{trivial} if one side of the cut has exactly one vertex.
A \emph{split} is a non-trivial mincut. An $s$-$t$ split is a split
that separates $s$ and $t$.

\subsection{An efficient split oracle for hypergraphs}
Given a hypergraph $H$ we would like to find a split if one exists.
It is not hard to come up with a polynomial-time algorithm for this
task but here we wish to design a faster algorithm. We accomplish
this by considering a weaker guarantee which suffices for our purposes.
The algorithm, given $H$ and the mincut value $\lambda$, outputs
either a split in $H$ or a pair of vertices $\{s,t\}$ such that
there is no $s$-$t$ split in $H$. We call such an algorithm a 
split oracle. We describe a near-linear-time split oracle. 

We first show how to use a maximum $s$-$t$ flow in $\vec{H}$ to help
decide whether there is an $s$-$t$ split, and compute one if there is.

\begin{lemma}
  Given a maximum $s$-$t$ flow in the equivalent digraph of $H$, and the
  value of mincut $\lambda$ in $H$, there is an algorithm that in
  $O(p)$ time either finds a $s$-$t$ split, or certifies that 
  no $s$-$t$ split exists in $H$.
\end{lemma}
\begin{proof}
  If the value of the maximum $s$-$t$ flow is greater than $\lambda$,
  there is no $s$-$t$ split. Otherwise, there is an $s$-$t$ split iff
  there is a non-trivial min-$s$-$t$ cut in $H$.

  Suppose a directed graph $G$ has $k$ minimum $u$-$v$-cuts for some vertex
  pair $(u,v)$. Given a maximum $u$-$v$ flow in $G$ and an integer $\ell$,
  there is an enumeration algorithm \cite{Provan1996} that outputs
  $\min(\ell, k)$ distinct min-$u$-$v$-cuts in $O(\ell m)$ time where
  $m$ is the number of edges in $G$.

  We run the enumeration algorithm with $\ell = 3$ on $\vec{H}$ for
  the pair $(s,t)$. Every min-$s$-$t$ cut in $\vec{H}$ corresponds to
  a min-$s$-$t$ cut in $H$.  If the algorithm returns at most two cuts
  and both are trivial then there is no $s$-$t$ split. Otherwise one
  of the output cuts is an $s$-$t$ split. The running time is $O(p)$
  since the number of edges in $\vec{H}$ is $O(p)$.
\qed\end{proof}

One can find a maximum $s$-$t$ flow in $\vec{H}$ using standard flow
algorithms but that would not lead to a near-linear time algorithm. In
graphs, Arikati and Mehlhorn \cite{Arikati1999} devised a linear-time
algorithm that computes the maximum flow between the last two vertices
of an MA-ordering. Thus, we have a near-linear-time split oracle for
graphs.  Recall that in hypergraphs there are three orderings which
all yield a pendant pair. We generalized Arikati and Mehlhorn's
algorithm to a linear-time algorithm that tries to find a maximum flow
between the last two vertices of an MA-ordering of a hypergraph (the
flow is in the equivalent digraph). Even though it appears to
correctly compute a maximum flow in all the experiments we ran, we
could not prove its correctness. Instead we found a different method
based on the tight ordering, that we describe below.

Let $v_1,v_2, \ldots, v_n$ be a tight ordering for a hypergraph
$H=(V,E)$. We define a tight graph $G=(V,E')$ with respect to $H$
and the given tight ordering as follows. For each edge $e\in E$, 
we add an edge $e'$ to $E'$, where $e'$ consists of the last $2$ vertices of 
$e$ under the tight ordering. The key observation is the following.

\begin{lemma}
  Suppose $H=(V,E)$ is a hypergraph and $v_1,\ldots,v_n$ is a tight
  ordering for $H$, and $G=(V,E')$ is the corresponding tight
  graph. Then, for $1 \le i<j \le n$, $d'(V_i,v_j;G) = d'(V_i,v_j;H)$.
\end{lemma}
\begin{proof}
  Consider any edge $e$ counted in $d'(V_i,v_j;H)$. $e\subset
  V_i\cup\set{v_j}$, and $e$ contains $v_j$. $e'$ contains $v_j$, and
  the other end of $e'$ is in $V_i$. Therefore $e'$ is counted
  in $d'(V_i,v_j;G)$. This shows that $d'(V_i,v_j;H)\leq d'(V_i,v_j;G)$.

  To see the other direction, consider an $e' \in E$ corresponding to
  an edge $e \in E$. If $e'$ is counted in $d'(V_i,v_j;G)$, it must be
  that $v_j$ is the last vertex in $e$ and the second last vertex of
  $e$ is in $V_i$. This implies that $e\subset V_i\cup \set{v_j}$, and
  therefore counted in $d'(V_i,v_j;H)$, and completes the direction
  $d'(V_i,v_j;H)\geq d'(V_i,v_j;G)$.
\qed\end{proof}

The preceding lemma implies that the tight ordering for $H$ 
is a tight ordering for $G$. From \autoref{lem:tightorderlasttwo},
\[
\lambda(v_{n-1},v_n;G) =d'(V_{n-1},v_n;G) =d'(V_{n-1},v_n;H)=\lambda(v_{n-1},v_n;H)
\]

Letting $s = v_{n-1}$ and $t= v_n$, we see that $\lambda(s,t;G) =
\lambda(s,t;H)$. Moreover, an $s$-$t$ flow in $G$ can be easily lifted
to an $s$-$t$ flow in $\vec{H}$. Thus, we can compute an $s$-$t$
max flow in $G$ in linear-time using the algorithm of
\cite{Arikati1999} and this can be converted, in linear time, into an
$s$-$t$ max-flow in $\vec{H}$.

This gives the following theorem.
\begin{theorem}
  \label{thm:split-oracle}
  The split oracle can be implemented in $O(p + n\log n)$ time for
  capacitated hypergraphs, and in $O(p)$ time for uncapacitated
  hypergraphs.
\end{theorem}

\subsection{Decompositions, Canonical and Prime}
We define the notion of decompositions to state the relevant theorem
on the existence of a canonical decomposition.  In later subsections
we describe the computational aspects.

A hypergraph $H$ is \emph{prime} if it does not contain any split; in
other words all mincuts of $H$ are trivial. A capacitated hypergraph
is called a \emph{solid polygon} if it consists of a cycle where each edge
has the same capacity $a$ and a hyperedge containing all vertices with
capacity $b$.  If $a=0$, it is called  \emph{brittle}, otherwise it is
called \emph{semi-brittle}. A solid polygon is \emph{not} prime if
it has at least $4$ vertices.  For a semi-brittle hypergraph with at
least $4$ vertices, every split consists of two edges on the cycle and
the hyperedge covering all vertices. For a brittle hypergraph with at
least $4$ vertices, any non-trivial cut is a split.

Given a hypergraph $H=(V,E)$ and a set $U$, a function $\phi:V\to U$
defines a new hypergraph through a sequence of contraction operations as
follows: for each element $u\in U$, contract $\phi^{-1}(u)$ into
$u$. The resulting hypergraph is the $\phi$-contraction of $H$. A
hypergraph obtained from $H=(V,E)$ by contracting $S \subset V$ into a
single vertex is denoted by $H/S$.

$\set{H_1,H_2}$ is a \emph{simple refinement} of $H$ if $H_1$ and
$H_2$ are hypergraphs obtained through a split $(V_1,V_2)$ of $H$ and
a new \emph{marker} vertex $x$ as follows.
\begin{enumerate}
\item $H_1$ is $H/V_2$, such that $V_2$ gets contracted to $x$.
\item $H_2$ is $H/V_1$, such that $V_1$ gets contracted to $x$.
\end{enumerate}

If $\set{H_1,H_2}$ is a simple refinement of $H$, then mincut value of
$H_1,H_2$ and $H$ are all equal.

A set of hypergraphs $\cD = \{H_1,H_2,\ldots,H_k\}$ is called a
\emph{decomposition} of a hypergraph $H$ if it is obtained from
$\set{H}$ by a sequence of operations each of which consists of
replacing one of the hypergraphs in the set by its simple refinement;
here we assume that each operation uses new marker vertices.  A
decomposition $\cD$ is a simple refinement of decomposition $\cD'$ if
$\cD$ is obtained through replacing one of the hypergraph in $\cD'$ by
its simple refinement. A decomposition $\cD'$ is a \emph{refinement}
of $\cD$ if $\cD'$ is obtained through a sequence of simple refinement
operations from $\cD$. If the sequence is non-empty, $\cD'$ is called a
\emph{strict refinement}. Two decompositions are \emph{equivalent} if they
are the same up to relabeling of the marker vertices. A decomposition
is \emph{minimal} with property $\mathcal{P}$ if it is not a strict
refinement of some other decomposition with the same property
$\mathcal{P}$. A \emph{prime} decomposition is one in which all
members are prime. A decomposition is \emph{standard} if every
element is either prime or a solid polygon.

Every element in the decomposition is obtained from a sequence of
contractions from $H$. Hence we can associate each element $H_i$ in
the decomposition with a function $\phi_{H_i}:V\to V(H_i)$, such that
$H_i$ is the $\phi_{H_i}$-contraction of $H$. Every decomposition
$\cD$ has an associated \emph{decomposition tree} obtained by having a node
for each hypergraph in the decomposition and an edge connecting two
hypergraphs if they share a marker vertex.

The important theorem below is due to \cite{Cunningham}, and stated
again in \cite{Cheng} specifically for hypergraphs.

\begin{theorem}[\cite{Cunningham}]
  Every hypergraph $H$ has a unique (up to equivalence) minimal standard
  decomposition. That is, any two minimal standard decompositions of $H$ differ
  only in the labels of the marker vertices.
\end{theorem}

The unique minimal standard decomposition is called the
\emph{canonical decomposition}. As a consequence, \emph{every}
standard decomposition is a refinement of the canonical decomposition.
We remark that minimality is important here. It captures all the
mincut information in $H$ as stated below.

\begin{theorem}[\cite{Cheng,Cunningham}]
\label{thm:cutmap}
Let $\cD=\set{H_1,\ldots, H_k}$ be a canonical decomposition of $H$.
\begin{enumerate}
\item For each mincut $S$ of $H$, there is a unique $i$, such that $\phi_{H_i}(S)$ is a mincut of $H_i$.
\item For each mincut $S$ of $H_i$, $\phi_{H_i}^{-1}(S)$ is a mincut of $H$.
\end{enumerate}
\end{theorem}

Note that each hypergraph in a canonical decomposition is either prime
or a solid polygon and hence it is easy to find all the mincuts in
each of them.  We observe that any decomposition $\cD$ of $H$ can be
compactly represented in $O(n)$ space by simply storing the vertex
sets of the hypergraph in $\cD$.

Recall that a set of edges $E' \subset E$ is called a min edge-cut-set
if $E' = \delta(S)$ for some mincut $S$. As a corollary of the preceding
theorem, one can easily prove that there are at most $n \choose 2$ distinct
min edge-cut-sets in a hypergraph. This fact is not explicitly stated 
in \cite{Cunningham,Cheng} or elsewhere in the literature but 
was known to those familiar with the decomposition theorem.
\begin{corollary}
A hypergraph with $n$ vertices has at most ${n\choose 2}$ distinct
min edge-cut-sets. 
\end{corollary}
\begin{proof}
  Let $H$ be a hypergraph on $n$ vertices.  If $H$ is prime, then
  there are at most $n$ min edge-cut-sets.  If $H$ is a solid-polygon,
  then there are at most ${n \choose 2}$ min edge-cut-sets.  Let $\cD$
  be a canonical decomposition of $H$. $\cD$ is obtained via a simple
  refinement $\set{H_1,H_2}$ of $H$ with size $a$ and $b$, followed by
  further refinement. Then, by induction, there are at most ${a\choose
    2}+{b\choose 2}$ min edge-cut-sets in $H_1$ and $H_2$.  Here
  $a+b=n+2$ and $a,b\leq n-1$. Therefore ${a\choose 2}+{b\choose
    2}\leq {3\choose 2}+{n-1\choose 2}\leq {n\choose 2}$ when $n\geq
  4$.
\qed\end{proof}

\subsection{Computing a canonical decomposition}
In this section we describe an efficient algorithm for computing the
canonical decomposition of a hypergraph $H$. 

We say that two distinct splits $(A, \bar{A})$ and $(B, \bar{B})$
cross if $A$ and $B$ cross, otherwise they do not cross.  One can
easily show that every decomposition is equivalently characterized by
the set of non-crossing splits induced by the marker vertices.
Viewing a decomposition as a collection of non-crossing splits
is convenient since it does not impose an order in which the splits
are processed to arrive at the decomposition --- any ordering of 
processing the non-crossing splits will generate the same decomposition.

Call a split \emph{good} if it is a split that is not crossed by any
other split; otherwise the split is called bad.  A canonical
decomposition corresponds to the collection of all good splits.  The
canonical decomposition can be obtained through the set of of good
splits \cite[Theorem 3]{CunninghamE} via the following simple
algorithm.  If $H$ is prime or solid polygon return $\{H\}$
itself. Otherwise find a good split $(A,\bar{A})$ and the simple
refinement $\{H_1,H_2\}$ of $H$ induced by the split and return the
union of the canonical decompositions of $H_1$ and $H_2$ computed
recursively. Unfortunately, finding a good split directly is
computationally intensive.

On the other hand finding a prime decomposition can be done via a
split oracle by a simple recursive algorithm, as we shall see in
Section~\ref{sec:prime-decomposition}.  Note that a prime
decomposition is not necessarily unique.  We will build a canonical
decomposition through a prime decomposition.  This was hinted in
\cite{Cunningham}, but without details and analysis. Here we formally
describe such an algorithm. 

One can go from a prime decomposition to a canonical decomposition by
removing some splits. Removing a split corresponds to gluing two
hypergraphs with the same marker vertex into another hypergraph
resulting in a new decomposition.  We formally define the operation as
follows.  Suppose $\cD$ is a decomposition of $H$ with a marker vertex
$x$ contained in $H_1$ and $H_2$. We define a new contraction of $H$
obtained by \emph{gluing} $H_1$ and $H_2$.  Let $\phi_{H_1}$ and
$\phi_{H_2}$ be the contractions of $H$, respectively.  Define
function $\phi':V\to (V(H_1)\cup V(H_2)) - x$ as follows
\[
\phi'(v) = \begin{cases}\phi_{H_1}(v) &\text{if }\phi_{H_2}(v)=x\\\phi_{H_2}(v) &\text{if }\phi_{H_1}(v)=x\end{cases} 
\]
$H_x$ is the contraction of $H$ defined by $\phi'$.  The gluing of
$\cD$ through $x$ is the set $\cD_x = \cD - \{H_1,H_2\} \cup \{H_x\}$.
The operation reflects removing the split induced by $x$ from 
the splits induced by $\cD$, therefore it immediately implies the following
lemma.

\begin{lemma}
  $\cD_x$ is a decomposition of $H$. Moreover, $\cD_x$ can be computed
  from $\cD$ and $H$ in $O(p)$ time.
\end{lemma}

\begin{remark}
  In order to compute $D_x$ implicitly, we only have to obtain
  $\phi_{H_1},\phi_{H_2}$ and compute a single contraction. Therefore
  $O(p)$ space is sufficient if we can obtain $\phi_{H_1}$ and
  $\phi_{H_2}$ in $O(p)$ time and space.
\end{remark}

We need the following simple lemma.

\begin{lemma}
\label{lem:decompofsolidpolygon}
  Let $H$ be a solid polygon. Any decomposition of $H$ is a standard
  decomposition. Thus, if $\cD$ is a decomposition of $H$, for any
  marker vertex, gluing it results in a standard decomposition of $H$.
\end{lemma}
\begin{proof}
  We first prove that any decomposition of $H$ is a standard
  decomposition.  This is by induction. If the solid polygon consists
  of a cycle with positive capacity, then exactly two edges in the cycle
  and the edge that contains all vertices crosses a split. One can
  verify that contraction of either side of the split results in a
  solid polygon or a prime hypergraph.  Otherwise, the solid polygon is a single
  hyperedge covering all vertices. Any contraction of this hypergraph
  is a solid polygon.

  The second part of the lemma follows from the first and the fact that
  gluing results in a decomposition.
\qed\end{proof}

The following lemma is easy to see. 
\begin{lemma}
  Given a hypergraph $H$ there is an algorithm to check if 
  $H$ is a solid polygon in $O(p)$ time.
\end{lemma}

Adding a split corresponds to a simple refinement. 
Therefore a decomposition $\cD'$ is a refinement of $\cD$ then
the set of induced splits of $\cD$ is a subset of induced splits of $\cD'$. 

Consider the following algorithm that starts with a prime
decomposition $\cD$. For each marker $x$, inspect if gluing through
$x$ results in a standard decomposition; one can easily check via the
preceding lemma whether the gluing results in a solid polygon which is
the only thing to verify. If it is, apply the gluing, if not, move on
to the next marker. Every marker will be inspected at most once,
therefore the algorithm stops after $O(n)$ gluing operations and takes
time $O(np)$. Our goal is to show the correctness of this simple
algorithm.

The algorithm starts with a prime decomposition $\cD$ which is a
standard decomposition. If it is minimal then it is canonical and no
gluing can be done by the algorithm (otherwise it would violate
minimality) and we will output a canonical decomposition as
required. If $\cD$ is not minimal then there is a canonical
decomposition $\cD^*$ such that $\cD$ is a strict refinement of
$\cD^*$. Let $\cD^* = \{H_1,H_2,\ldots, H_k\}$ where each $H_i$ is
prime or a solid polygon. Therefore $\cD = \cup_{i=1}^k \cD_i$ where
$\cD_i$ is a refinement of $H_i$. If $H_i$ is prime than $\cD_i =
\{H_i\}$. If $H_i$ is a solid polygon then $\cD_i$ is a standard
decomposition of $H_i$. Our goal is to show that irrespective of the
order in which we process the markers in the algorithm, the output will
be $\cD^*$. Let the marker set for $\cD^*$ be $M^*$ and that for
$\cD$ be $M \supset M^*$. Suppose the first marker considered by
the algorithm is $x$. There are two cases.

The first case is when $x \in M - M^*$. In this case the marker $x$
belongs to two hypergraphs $G_1$ and $G_2$ both belonging to some
$\cD_i$ where $H_i$ is a solid polygon. The algorithm will glue $G_1$
and $G_2$ and from \autoref{lem:decompofsolidpolygon}, this gives a
smaller standard decomposition $\cD'_i$ of $H_i$.

The second case is when the marker $x \in M^*$. Let $x$ belong to two
hypergraphs $G_1$ and $G_2$ where $G_1 \in \cD_i$ and $G_2 \in \cD_j$
where $i \neq j$. In this case we claim that the algorithm will not
glue $G_1$ and $G_2$ since gluing them would not result in a standard
decomposition. To see this, let $\cD'$ be obtained through gluing of
$G_1$ and $G_2$. The split induced by $x$ is in $\cD^*$ but not
$\cD'$. Because the splits induced by $\cD^*$ is not a subset of
splits induced by $\cD'$, $\cD'$ is not a refinement of $\cD^*$.
However, as we noted earlier, every standard decomposition is
a refinement of $\cD^*$. Hence $\cD'$ is not a standard decomposition. 

From the two cases we see that no marker in $M^*$ results in a gluing
and every marker in $M-M^*$ results in a gluing. Thus the algorithm
after processing $\cD$ outputs $\cD^*$. This yields the following
theorem.

\begin{theorem}
\label{thm:candecompfromprimedecomp}
A canonical decomposition can be computed in $O(np)$ time given a
prime decomposition.
\end{theorem}

Next we describe an $O(np+n^2\log n)$ time algorithm to compute a prime decomposition.

\subsection{Computing a prime decomposition}
\label{sec:prime-decomposition}
We assume there exists an efficient \emph{split oracle}. Given a hypergraph $H$
and the value of the mincut $\lambda$, the split oracle finds a split
in $H$ or returns a pair $\set{s,t}$, such that there is no $s$-$t$ split
in $H$. In the latter case we would like to recurse on the hypergraph
obtained by contracting $\set{s,t}$ into a single vertex. In order
to recover the solution, we define how we can uncontract the contracted 
vertices.

\begin{figure}[htb]
\begin{algo}
\textsc{Prime}$(H,\lambda)$\+
\\   if $|V(H)|\ge 4$\+
\\      $x \gets$ a new marker vertex
\\      query the split oracle with $H$ and $\lambda$
\\      if oracle returns a split $(S,V(H)-S)$\+
\\          $\set{H_1,H_2} \gets \textsc{Refine}(H,(S,V(H)-S),x)$
\\          return $\textsc{Prime}(H_1,\lambda) \cup \textsc{Prime}(H_2,\lambda)$\-
\\      else the oracle returns $\set{s,t}$\+
\\          $\cD'\gets \textsc{Prime}(H/\set{s,t},\lambda)$, $\set{s,t}$ contracts to $v_{\set{s,t}}$
\\          $G' \gets$ the member of $\cD'$ that contains $v_{\set{s,t}}$
\\          $G \gets $ uncontract $v_{\set{s,t}}$ in $G'$ with respect to $H$
\\          if $(\set{s,t},V(G)-\set{s,t})$ is a split in $G$\+
\\              $\set{G_1,G_2}\gets$ refinement of $G$ induced by $\set{s,t}$
\\              $\cD\gets (\cD' - \set{G'})\cup\set{G_1,G_2}$   \-
\\          else\+
\\              $\cD\gets (\cD' - \set{G'})\cup G$\-
\\          return $\cD$ \-\-
\\   else\+
\\      return $\set{H}$
\-
\end{algo}
\caption{The algorithm for computing a prime decomposition.}
\label{alg:prime}
\end{figure}

\begin{definition}
  Consider a hypergraph $H$. Let $H' = H/\set{s,t}$, where $\set{s,t}$
  is contracted to vertex $v_{\set{s,t}}$.  Let $G'$ be a
  $\phi'$-contraction of $H'$ such that
  $\phi'(v_{\set{s,t}})=v_{\set{s,t}}$.  We define uncontracting
  $v_{\set{s,t}}$ in $G'$ with respect to $H$ as a graph $G$ obtained
  from a $\phi$-contraction of $H$, where $\phi$ is defined as
 
\[
\phi(v) = \begin{cases} \phi'(v) & \text{if } v\not\in \set{s,t}\\
                          v       & otherwise\end{cases}
\]
\end{definition}

See \autoref{alg:prime} for a simple recursive 
algorithm that computes a prime decomposition based on the split oracle.
The following lemma justifies the soundness of recursing
on the contracted hypergraph when there is no $s$-$t$ split.

\begin{lemma}
\label{lem:contractst}
Suppose $H$ is a hypergraph with no $s$-$t$ split for some $s,t \in V(H)$.
Let $H' = H/\set{s,t}$, where $\set{s,t}$ is contracted to vertex
$v_{\set{s,t}}$.  Let $\cD'$ be a prime decomposition of $H'$, 
and let $G' \in \cD'$ such that $G'$ contains vertex $v_{\set{s,t}}$. 
And let $G$ be obtained through uncontracting $v_{\set{s,t}}$ in $G'$ 
with respect to $H$.

\begin{enumerate}
\item Suppose $\set{s,t}$ defines a split in $G$ and
let $\set{G_1,G_2}$ be a simple refinement of $G$ based on this split. 
Then $\cD = (\cD' -\set{G'}) \cup \set{G_1,G_2}$ is a prime decomposition of $H$.
\item If  $\set{s,t}$ does not define a split in $G$ then
  $\cD = (\cD' -\set{G'})\cup \set{G}$ is a prime decomposition of $H$.
\end{enumerate}
\end{lemma}
\begin{proof}
  Every split in $H'$ is a split in $H$.  Therefore $(\cD'-
  \set{G'})\cup \set{G}$ is a decomposition of $H$. Other than $G$, all other
  elements in $(\cD'-
  \set{G'})\cup \set{G}$ are prime. 

  If $G$ is not prime, then there is a split. There is no $s$-$t$ split
  in $G$ because $H$ does not have any $s$-$t$ split. Any split in $G$
  must have the form $(A,V(G)-A)$ where $\set{s,t}\subset A$. If
  $A\neq \set{s,t}$, then there exist some other vertex $v\in A$,
  which implies $|A-\set{s,t}\cup \set{v_{\set{s,t}}}|\geq 2$, and
  $(A-\set{s,t}\cup \set{v_{\set{s,t}}},V(G')-A)$ is a split in $G'$,
  a contradiction to the fact that $G'$ is prime. Hence
  $(\set{s,t},V(G) - \set{s,t})$ is the unique split in $G$.
  Therefore the simple refinement of $G$ based on this unique split
  are both prime, and we reach the first case.

  If $G$ is prime, then we are done, as we reach the second case. 
\qed\end{proof}

\begin{theorem}
\textsc{Prime}$(H,\lambda)$ outputs a prime decomposition in $O(n(p+T(n,m,p)))$ time. Where $T(n,m,p)$ is the time to query split oracle with a hypergraph of $n$ vertices, $m$ edges and sum of degree $p$. 
\end{theorem}
\begin{proof}
Using induction and \autoref{lem:contractst}, the correctness of the algorithm is clear. \textsc{Prime} is called at most $2n$ times, and each call takes $O(p+T(n,m,p))$ time. 
\qed\end{proof}

Using the split oracle from \autoref{thm:split-oracle} we obtain the following corollary.
\begin{corollary}
  \label{cor:primedecomp}
  A prime decomposition of a capacitated hypergraph can be computed in
  $O(np + n^2 \log n)$ time. For uncapacitated hypergraphs it can be
  computed $O(np)$ time.
\end{corollary}

\subsection{Reducing space usage}
\label{sec:space}
Our description of computing the prime and canonical decompositions did not
focus on the space usage. A naive implementation can use $\Omega(np)$ space
if we store each hypergraph in the decomposition explicitly. Here we
briefly describe how one can reduce the space usage to $O(p)$ by
storing a decomposition implicitly via a \emph{decomposition tree}.

Consider a decomposition $\cD=\set{H_1,\ldots,H_k}$ of $H=(V,E)$.  We
associate a decomposition tree $T=(A,F)$ with $\cD$ where $A =
\{a_1,\ldots,a_k\}$, one node per hypergraph in $\cD$; there is an
edge $a_ia_j \in F$ iff $H_i$ and $H_j$ share a marker vertex.  With
each $a_i$ we also store $V(H_i)$ which includes the marker vertices
and some vertices from $V(H)$. This is stored in a map $\psi : A
\rightarrow \cup_i V(H_i)$.  It is easy to see that the total storage
for the tree and storing the vertex sets is $O(n)$; a marker vertex
appears in exactly two of the hypergraphs of a decomposition and
a vertex of $H$ in exactly one of the hypergraphs in the decomposition.

Given the decomposition tree $T$ and $\psi$ and a node $a_i \in A$, we
can recover the hypergraph $H_i$ (essentially the edges of $H_i$ since
we store the vertex sets explicitly) associated with a node $a_i$ in
$O(p)$ time. For each edge $e$ incident to $a_i$ in $T$, let $C_e$ be
the component of $T-e$ that does not contain $a_i$. $V(H) \cap
(\cup_{a_j \in C_e} \psi(a_j))$ are the set of vertices in $H$ which
are contracted to a single marker vertex in $H_i$ corresponding to the
edge $e$.  We collect all this contraction information and then apply
the contraction to the original hypergraph $H$ to recover the edge set of
$H_i$. It is easy to see that this can be done in $O(p)$ time.

\subsection{Hypercactus representation}
For a hypergraph $H$, a hypercactus representation is a hypercactus
$H^*$ and a function $\phi:V(H) \to V(H^*)$ such that for all
$S\subset V(H)$, $S$ is a mincut in $H$ if and only if $\phi(S)$ is a
mincut in $H^*$. This is a generalization of the
cactus representation when $H$ is a graph.

Note the similarity of \autoref{thm:cutmap} and the definition of the
hypercactus representation. It is natural to ask if there is a
hypercactus representation that is essentially a canonical
decomposition. Indeed, given the canonical decomposition of $H$, Cheng
showed that one can construct a ``structure hypergraph'' that captures
all mincuts~\cite{Cheng}, which Fleiner and Jordan later point
out is a hypercactus representation \cite{FleinerJ}. The process to
construct such a hypercactus representation from a canonical
decomposition is simple. We describe the details for the sake of completeness.

Assume without loss of generality that $\lambda(H) = 1$.  We construct
a hypercactus if the hypergraph is prime or a solid polygon.  If $H$
is a solid polygon, then it consists of a cycle and a hyperedge
containing all the vertices. If the cycle has non-zero capacity, let
$H^*$ to be $H$ with the hyperedge containing all the vertices
removed, and assign a capacity of $\frac{1}{2}$ to each edge of the
cycle.  If the cycle has zero capacity, then let $H^*$ to be a single
hyperedge containing all vertices, the hyperedge has capacity $1$. In
both cases $H^*$ together with the the identity function on $V(H)$
forms a hypercactus representation for $H$. If $H$ is prime, let $V'$
be the set of vertices that induce a trivial mincut, i.e. $v\in V'$
iff $\set{v}$ is a mincut in $H$.  Introduce a new vertex $v_H$, and
let $H^* = (\set{v_H} \cup V', \set{ \set{v_H,v'} | v'\in V'})$, with
 capacity $1$ for each edge; in other words we create a star with
center $v_H$ and leaves in $V'$.  Define $\phi:V(H)\to V(H^*)$ as
\[
\phi(u) =\begin{cases}
u & u\in V'\\
v_H & u\not\in V'
\end{cases}
\]
Then $H^*$ and $\phi$ form a hypercactus representation. 

For the more general case, let $\cD^*=\set{H_1,\ldots,H_k}$ be the
canonical decomposition of $H$. For each $i$, construct hypercactus
representation $(H^*_i,\phi_i)$ of $H_i$ as described earlier.  We
observe that that if $x$ is a marker vertex in $H_i$, then it is also
present in $H^*_i$. If $H_i$ is a solid polygon this is true because
$V(H_i)=V(H^*_i)$. If $H_i$ is prime, then every marker vertex induces
a trivial mincut in $H_i$, hence also preserved in $H^*_i$.  Construct
$H^*$ from $H^*_1,\ldots,H^*_k$ by identifying marker vertices.
This also gives us $\phi: V(H) \to V(H^*)$ by gluing together $\phi_1,\ldots,\phi_k$ naturally. $(H^*,\phi)$ is the desired hypercactus representation.

The construction takes $O(np)$ time and $O(p)$ space.

\begin{theorem}
\label{thm:buildcactus}
A hypercactus representation of a capacitated hypergraph can be found in
$O(n(p+n\log n))$ time and $O(p)$ space for capacitated hypergraphs,
and in $O(np)$ time for uncapacitated hypergraphs.
\end{theorem}
\begin{proof}
We combine \autoref{thm:candecompfromprimedecomp}
and \autoref{cor:primedecomp}. The space usage can be
made $O(p)$ based on the discussion in Section~\autoref{sec:space}.
\qed\end{proof}

If $H$ is a graph, the hypercactus representation constructed is a
cactus representation. \autoref{thm:buildcactus} matches the best
known algorithm for cactus representation construction of graphs in
both time and space \cite{Nagamochi2003}, and is conceptually simpler.

Via sparsification we obtain a faster algorithm for uncapacitated hypergraphs.

\begin{theorem}
  A hypercactus representation of an uncapacitated hypergraph can be
  found in $O(p+\lambda n^2)$ time and $O(p)$ space.
\end{theorem}
\begin{proof}
  Find the mincut value $\lambda$, and a $(\lambda+1)$-sparsifier $H'$
  of $H$ in $O(p+\lambda n^2)$ time. \autoref{thm:sparsification-cuts}
  shows that every mincut in $H$ is a mincut
  in $H'$, and vice versa. Therefore the hypercactus for $H'$ is a
  hypercactus for $H$. Apply \autoref{thm:buildcactus} to $H'$.
\qed\end{proof}

\section{Near-linear time $(2+\e)$ approximation for mincut}
\label{sec:approx-mincut}
Matula gave an elegant use of MA-ordering to obtain a
$(2+\e)$-approximation for the mincut in an uncapacitated undirected
graph in $O(m/\e)$ time~\cite{Matula1993}. Implicit in his paper is an
algorithm that gives a $(2+\e)$-approximation for capacitated graphs
in $O(\frac{1}{\e}(m \log n +n \log^2 n))$ time; this was explicitly
pointed out by Karger~\cite{Karger1995}. Here we extend Matula's idea
to hypergraphs. We describe an algorithm that outputs a
$(2+\e)$-approximation for hypergraph mincut in $O(\frac{1}{\e}(p \log
n + n \log^2 n))$ time for capacitated case, and in $O(p/\e)$ time for
the uncapacitated case.

We will assume without loss of generality that there is no edge
that contains all the nodes of the given hypergraph.
Let $v_1,\ldots,v_n$ be a MA-ordering of the given capacitated hypergraph
$H$. Given a non-negative number $\alpha$, a set of consecutive
vertices in the ordering $v_a,v_{a+1},\ldots,v_b$ where $a \le b$ is
called $\alpha$-tight if $d(V_i,v_{i+1})\geq \alpha$ for all $a\leq
i<b$. The {\em maximal} $\alpha$-tight sets partition $V$.  We obtain
a new hypergraph by contracting each maximal $\alpha$-tight set into a single
vertex. Edges that become singletons in the contraction are discarded.
We call the contracted hypergraph an $\alpha$-contraction.
Note that the contraction depends both on $\alpha$ and the specific
MA-ordering.

One important aspect of $\alpha$-contraction is that the resulting
hypergraph has sum-degree at most $2\alpha n$ which allows for sparsifying
the hypergraph by appropriate choice of $\alpha$.

\begin{lemma}
\label{lem:contractsize}
Let $H'$ be an $\alpha$-contraction of a given hypergraph $H$.
Then $\sumdeg(H') \le 2 \alpha n$ where $n$ is the number of nodes of $H$.
\end{lemma}

\begin{proof}
  Assume the $\alpha$-tight partition of $V$ is $X_1,\ldots,X_h$ where the
  ordering of the parts is induced by the MA-ordering.  For
  $1 \le i \le h$, let $A_i =\bigcup_{j=1}^i X_j$. Since each $X_i$ is
  a maximal $\alpha$-tight set, $d(A_i,x)<\alpha$ for all $x\in
  X_{i+1}$. Let $E'$ be the set of edges in $H'$. For a given edge $e \in H$
  let $e'$ be the corresponding edge in $H'$. Note that $|e'| \ge 2$ and
  $c(e') = c(e)$ since the capacity is unchanged. We have the following
  set of inequalities:
  \begin{align*}
    \sumdeg(H')  &= \sum_{e' \in E'} c(e') |e'|\\
                 &\le  2\sum_{e' \in E'} c(e') (|e'|-1) \\
                 &= 2\sum_{i=1}^{h-1} d(A_i, X_{i+1}) \\
                 &\le  2\sum_{i=1}^{h-1} \sum_{x \in X_{i+1}} d(A_i, x) \\
                 &< 2 \sum_{i=1}^{h-1} \alpha |X_{i+1}| \le 2 \alpha n. 
  \end{align*}
\qed\end{proof}

The second important property of $\alpha$-contraction 
is captured by the following lemma.
\begin{lemma}
\label{lem:alphatight}
If $v_i$ and $v_j$ are in a $\alpha$-tight set
then $\lambda(v_i,v_j)\geq \alpha$.
\end{lemma}
\begin{proof}
  Assume without loss of generality that $i<j$.  Consider any $k$ such
  that $i \le k < j$.  We have $d(V_k,v_{k+1}) \ge \alpha$ because $i$
  and $j$ are in the same $\alpha$-tight set. By
  \autoref{lem:connectivityMAorder}, $\lambda(v_k,v_{k+1})\geq
  d(V_k,v_{k+1}) \ge \alpha$. By induction and using the fact that for
  any $a,b, c \in V$, $\lambda(a,c)\geq
  \min(\lambda(a,b),\lambda(b,c))$, we have $\lambda(v_i,v_j) \geq
  \alpha$.
\qed\end{proof}

\autoref{fig:approx-mincut} describes a simple recursive algorithm for
finding an approximate mincut.

\begin{figure}
\begin{algo}
\textsc{Approximate-MinCut}$(H)$\+
\\   if ($|V(H)|\geq 2$)\+
\\       $\delta = \min_{v\in V} \delta(v)$
\\       if $(\delta = 0)$ \+
\\          return $0$ \-
\\       $\alpha \gets \frac{1}{2+ \e}\delta$
\\       Compute MA-ordering of $H$ 
\\       $H'\gets $ $\alpha$-contraction of $H$ 
\\       $\lambda' \gets$ \textsc{Approximate-MinCut}($H'$)
\\       return  $\min(\delta,\lambda')$ \-
\\   else\+
\\      return $\infty$
\-
\end{algo}
  \caption{Description of $(2+\e)$-approximation algorithm. It is easy
to remove the recursion.}
  \label{fig:approx-mincut}
\end{figure}

\begin{theorem}
$\textsc{Approximate-MinCut}$ outputs a $(2+\e)$-approximation to
an input hypergraph $H$ and can be implemented in
$O(\e^{-1} (p + n \log n) \log \frac{ n \delta(H)}{\lambda (H)}))$ time 
for capacitated hypergraphs and in $O(\e^{-1}p)$ time for uncapacitated 
hypergraphs. 
\end{theorem}

\begin{proof}
  We first argue about the termination and run-time. From
  \autoref{lem:contractsize}, $\sumdeg(H') \le \frac{2}{2+\e} n
  \delta(H)$.  Since $\sumdeg(H) \ge n \delta(H)$, we see that each
  recursive call reduces the sum degree by a factor of
  $\frac{2}{2+\e}$.  This ensures termination. 

  If the hypergraph is uncapacitated, the running time of each
  iteration is bounded by the sum degree. The sum degree reduces by a
  factor of $\frac{2}{2+\e}$ in each iteration. Hence $O(\sum_i^\infty
  (2/(2+\e))^i p)=O(\e^{-1}p)$ upper bounds the running time for
  uncapacitated hypergraphs.

  We now consider the case when the hypergraph is capacitated.  Let
  $H''$ be any non-trivial hypergraph (that has at least two vertices)
  that arises in the recursion.  The mincut value does not
  reduce by contraction and hence $\lambda(H'') \ge \lambda(H)$ which
  in particular implies that $\delta(H'') \ge \lambda(H)$, and hence
  $\sumdeg(H'') \ge 2 |V(H'')| \lambda(H)$.  After the first recursive
  call the sum degree is at most $\frac{2}{2+\e} n \delta(H)$. Thus
  the total number of recursive calls is $O(\e^{-1}\log(\frac{n
    \delta(H)}{\lambda(H)}))$. The work in each call is dominated by
  the time to compute an MA-ordering which can be done in $O(p + n
  \log n)$. This time gives the desired upper bound on the run-time of
  the algorithm for capacitated hypergraphs.

  We now argue about the correctness of the algorithm which is by
  induction on $n$. It is easy to see that the algorithm correctly
  outputs the mincut value if $n = 1$ or if $\delta = 0$. Assume
  $n \ge 2$ and $\delta(H) > 0$. The number of vertices in $H'$ is
  strictly less than $n$ if $\delta(H) > 0$ since the sum degree
  strictly decreases. Since contraction does not reduce the minimum
  cut value, $\lambda(H') \ge \lambda(H)$. By induction, $\lambda(H')
  \le \lambda' \le (2+\e) \lambda(H')$.  If $\lambda(H') =
  \lambda(H)$ then the algorithm outputs a $(2+\e)$-approximation
  since $\delta \ge \lambda(H)$.  The more interesting case is if
  $\lambda(H') > \lambda(H)$. This implies that there are two distinct
  nodes $x$ and $y$ in $H$ such that $\lambda(x,y,H) = \lambda(H)$ and
  $x$ and $y$ are contracted together in the $\alpha$-contraction. By
  \autoref{lem:alphatight}, $\lambda(x,y,H) \ge \alpha =
  \frac{1}{2+\e}\delta$ which implies that $\delta \le (2+\e)
  \lambda(H)$. Since the algorithm returns $\min(\delta, \lambda')$ we
  have that the output is no more than $(2+\e)\lambda(H)$. 

\qed\end{proof}

Since $\delta$ can be much larger than $\lambda$ in a capacitated hypergraph, 
we can preprocess the hypergraph to reduce $\delta$ to at most $n\lambda$
to obtain a strongly polynomial run time.

\begin{lemma}
\label{thm:reducecapacity}
Let $\beta = \min_{i > 1} d(V_{i-1},v_i)$ for a given
MA-ordering $v_1,\ldots,v_n$ of a capacitated hypergraph $H$. 
Then $\beta \le \lambda(H) \le n \beta$.
\end{lemma}
\begin{proof}
  From \autoref{lem:alphatight}, $\lambda(u,v)\geq \beta$ for all
  $u,v\in V$ because $V$ is a $\beta$-tight set. Therefore 
  $\lambda(H) \ge \beta$.  Let $i^* = \argmin_{i > 1} d(V_{i-1},v_i)$. Then,
\[
d(V_{i^*-1},V\setminus V_{i^*-1}) \leq \sum_{j=i^*}^n d(V_{i^*-1},v_j) \leq (n+1-i^*) \beta \leq n\beta.
\]
Thus, the cut $(V_{i^*-1}, V\setminus V_{i^*-1})$ has capacity 
at most $n\beta$, and hence $\lambda(H) \le n \beta$.
\qed\end{proof}

Let $\beta$ be the value in \autoref{thm:reducecapacity}, then a
$2n\beta$-contraction of $H$ yields a non-trivial hypergraph $H'$
where $\sumdeg(H') = O(n^2\beta)$. This also implies that
$\delta(H') = O(n^2 \beta)$.  Applying the
$(2+\e)$ approximation algorithm to $H'$ gives us the
following corollary. 

\begin{corollary}
\label{cor:2eapproximation}
A $(2+\e)$ approximation for hypergraph mincut 
can be computed in $O(\e^{-1} (p + n\log n)\log n)$ time for capacitated
hypergraphs, and in $O(\e^{-1}p)$ time for uncapacitated hypergraphs.
\end{corollary}

\begin{remark}
  Suppose we use the Queyranne ordering instead of MA-ordering, and
  define $v_a,\ldots,v_b$ to be $\alpha$-tight if
  $\frac{1}{2}(d(V_i,v_{i+1})+d'(V_i,v_{i+1}))\geq \alpha$ for $a\leq
  i<b$. The algorithm in \autoref{fig:approx-mincut} produces a
  $(2+\e)$-approximation with this modification.
\end{remark}

\section{Concluding Remarks}
We close with some open problems. The main one is to find an algorithm
for hypergraph mincut that is faster than the current one that runs in
$O(np + n^2 \log n)$ time. We do not know a better deterministic
run-time even when specialized to graphs. However we have a randomized
near-linear time algorithm for graphs \cite{Karger00}.  Can Karger's
algorithm be extended to hypergraphs with fixed rank $r$? Recently
there have been several fast $s$-$t$ max-flow algorithms for
undirected and directed graphs. The algorithms for directed graphs
\cite{LeeS14, Madry13} have straight forward implications for hypergraphs
$s$-$t$ cut computation via the equivalent digraph.  However,
hypergraphs have additional structure and it may be feasible to find
faster (approximate) algorithms.

We described a linear-time algorithm to find a maximum-flow between
the last two vertices of a tight-ordering of a hypergraph (the flow is
in the equivalent digraph of the hypergraph). We believe that such a
linear-time algorithm is also feasible for the last two vertices of an
MA-ordering of a hypergraph. Some of the research in this paper was
inspired by work on element connectivity and we refer the reader to
\cite{Chekuri-ec-survey15} for related open problems.

\paragraph{Acknowledgments:} We thank Yosef Pogrow for pointing out a
flaw in the proof of \autoref{thm:sparsification-cuts} in a
previous version of the paper. We also like to thank Tao Du for pointing
out a issue in sentences leading up to \autoref{cor:2eapproximation}.

\bibliographystyle{plain} 
\bibliography{hypergraph}
\end{document}